\documentclass[journal,twoside,web]{ieeecolor}
\usepackage[english]{babel}

\usepackage{generic}
\usepackage{cite}

\usepackage{amsmath,amssymb,amsthm} % assumes amsmath package installed
\usepackage{epsfig}
\usepackage{subfigure}
\usepackage{epstopdf}
\usepackage{tikz}
\usepackage{algorithm,algorithmicx,algpseudocode}% for pseudocode
\let\labelindent\relax
\usepackage{enumitem}
\usepackage{url}
\usepackage[nolist]{acronym}

\newacro{qp}[QP]{Quadratic Program}
\newacro{cbf}[CBF]{Control Barrier Function}
\newacro{zbf}[ZBF]{Zeroing Barrier Function}
\newacro{zcbf}[ZCBF]{Zeroing Control Barrier Function}
\newacro{ds}[DS]{Dynamical System}

\usepackage[makeroom]{cancel}
\floatstyle{ruled}
\newfloat{algorithm}{tbp}{loa}
\providecommand{\algorithmname}{Algorithm}
\floatname{algorithm}{\protect\algorithmname}

\makeatletter
\newcommand\footnoteref[1]{\protected@xdef\@thefnmark{\ref{#1}}\@footnotemark}

\makeatother

\renewcommand{\algorithmiccomment}[1]{\bgroup\hfill\scriptsize//~#1\egroup}

\newcommand{\tr}{{^{\top}}}

\newtheorem{theorem}{Theorem}
\newtheorem{proposition}{Proposition}
\newtheorem{definition}{Definition}

\DeclareMathOperator*{\minimize}{minimize~}

\DeclareMathOperator*{\st}{subject\,to~}

\newcommand{\R}{\mathbb R}
\newcommand{\mc}{\mathcal} % Bold letters, Greek alphabet, ...

\definecolor{my_green}{rgb}{0.0,0.49,0.19}

\newcommand{\changed}[1]{{\color{black}#1}}

\graphicspath{{img/}}

\begin{document}

\title{Safety of Dynamical Systems with Multiple Non-Convex Unsafe Sets\\Using Control Barrier Functions
}
\author{Gennaro Notomista, \IEEEmembership{Member, IEEE} and Matteo Saveriano, \IEEEmembership{Member, IEEE}
\thanks{ 
\textcopyright 2021 IEEE.  Personal use of this material is permitted.  Permission from IEEE must be obtained for all other uses, in any current or future media, including reprinting/republishing this material for advertising or promotional purposes, creating new collective works, for resale or redistribution to servers or lists, or reuse of any copyrighted component of this work in other works.}
\thanks{This work has been partially supported by the Austrian Research Foundation (Euregio IPN 86-N30, OLIVER)}
\thanks{G. Notomista is with CNRS, University of Rennes, Inria, IRISA, Rennes, France. {\tt{gennaro.notomista}@irisa.fr}}
\thanks{M. saveriano is with Department of Computer Science and Digital Science Center (DiSC), University of Innsbruck, Innsbruck, Austria. {\tt matteo.saveriano@uibk.ac.at}.}
}

\pagestyle{empty}
\maketitle
\thispagestyle{empty}

\begin{abstract}
	This paper presents an approach to deal with safety of dynamical systems in presence of multiple non-convex unsafe sets. While optimal control and model predictive control strategies can be employed in these scenarios, they suffer from high computational complexity in case of general nonlinear systems. Leveraging control barrier functions, on the other hand, results in computationally efficient control algorithms. Nevertheless, when safety guarantees have to be enforced alongside stability objectives, undesired asymptotically stable equilibrium points have been shown to arise. We propose a computationally efficient optimization-based approach which allows us to ensure safety of dynamical systems without introducing undesired equilibria even in presence of multiple non-convex unsafe sets. The developed control algorithm is showcased in simulation and in a real robot navigation application.
\end{abstract}


\IEEEpeerreviewmaketitle

%===================== INTRODUCTION ===========================================%
\section{Introduction}\label{sec:intro}
\IEEEPARstart{T}{he} study of safety of dynamical systems is rapidly developing thanks to new theoretical and computational tools which are now available, and which allow us to tackle fundamental problems in a large variety of disciplines, ranging from fast computation of collision-free robot motion \cite{frazzoli2002real} to safe learning \cite{aswani2013provably,saveriano2019learning}. Given a state space representation of a dynamical system, safety can be defined as the forward invariance condition of a subset of the state space, i.e., a system is safe with respect to a set $\mc S$ if the trajectory of the state, $x(t)$, satisfies $x(t_0)\in\mc S \implies x(t)\in\mc S \quad\forall t\ge t_0$.

Applications of safety can be found numerous in various research fields. These include robotics, where safe planning \cite{petti2005safe}, navigation \cite{loizou2017navigation,thyri2020reactive}, and autonomy \cite{jha2018safe,notomista2020long} have been extensively studied; safe reinforcement learning \cite{perkins2002lyapunov,cheng2019end} in machine learning; and, more recently, epidemiology \cite{molnar2020safety}. %, where control of infections has been encoded, enforced, and achieved by considering the forward invariance property of a safe (healthy) set of the system state space.
While planning-like strategies, in the form of optimal and model predictive control, have been developed to deal with safety objectives (see, e.g., \cite{aswani2013provably}), \acp{cbf} \cite{ames2019control}, are gaining popularity thanks to their low computational complexity and their ability to encode a rich variety of safety specifications. % (see \cite{ames2019control} and references therein). 
Nevertheless, the low computational complexity comes at the cost of having a reactive control formulation, as opposed to planning-like strategies. 
 
Reis et al.~\cite{reis2020control} have shown how the reactive nature of \acp{cbf} and the concurrent presence of competing objectives, namely stability and safety, can generate undesired and asymptotically stable equilibrium points. This phenomenon is particularly critical as undesirable asymptotically stable equilibria exist even in presence of convex unsafe regions. In \cite{reis2020control}, the authors propose a solution based on the transformation of the control Lyapunov function used to achieve the stability objective, which works in presence of a single convex unsafe set.

In this paper, we present an approach which solves the problem of undesirable asymptotically stable equilibrium point in the case of multiple and non-convex unsafe regions. The strategy consists of mapping the \textit{real world}, i.e., the system state space with non-convex unsafe regions, to a \textit{ball world}, where unsafe regions are either closed balls or the complement of open balls. Then, in the ball world, which is going to be defined as the image of the real world through a diffeomorphism, we compute a controller to transform the balls by  changing  their location  and  size, so that safety specification on the system in the real world are satisfied. The resulting algorithm consists of a convex quadratic program, which is agnostic of the method used to guarantee stability objectives---i.e., it is not based on the knowledge of the control Lyapunov function used to stabilize the system as done in~\cite{reis2020control}.

\changed{To summarize, the main contributions of this paper are: (i) We present a \ac{cbf}-based control strategy that mitigates the problem of undesirable asymptotically stable equilibrium points for the case of multiple non-convex unsafe regions; (ii) We provide a computationally-efficient optimization-based strategy that implements the proposed solution; (iii) We show the effectiveness of the proposed approach using both simulations and a real robotic experiment of a navigation task.}

%===================== RELATED WORK ===========================================%
% \input{Sections/RelatedWork.tex}

%===================== SECTION II ===========================================%
\section{Background}\label{sec:background}
In this paper, we employ \acp{cbf} to synthesize constraints on the control input to the obstacles which guarantee the safety, intended as the forward invariance, of a safe set in the state space. \changed{In the following, we recall the definition of \acp{cbf} and an optimization-based approach to synthesize safe controllers.}

\subsection{Control Barrier Functions}\label{subsec:cbf}
Consider the control affine system
\begin{equation}
    \dot{x} = f(x) + g(x)u\label{eq:ds_syst_cont},
\end{equation}
where $x \in \mc X \subset \mathbb{R}^n$ is the state, $\mc X$ compact, and $u \in \mathbb{R}^m$ is the input of the system. $f: \mathbb{R}^{n} \rightarrow \mathbb{R}^{n}$  and $g: \mathbb{R}^{n} \rightarrow \mathbb{R}^{n\times m}$ are locally Lipschitz vector fields. Assume we are interested in keeping the state within a safe set $\mc S \subset \mc D \subset \R^n$ which can be expressed as the zero-superlevel set of a continuously differentiable function $h: \mathcal{D} \rightarrow \mathbb{R}$ as follows:
\begin{equation}
\mc S = \{x \in  \mathcal{D} : h(x) \geq 0 \} \label{eq:c_set}.
\end{equation}

\begin{definition}[\acp{cbf} \cite{ames2019control}]
	\label{def:cbf}
	Let $\mc S \subset \mc D \subset \R^n$ be the zero superlevel set of a continuously differentiable function $h\colon\mc D \to \R$. Then $h$ is a control barrier function (CBF) if there exists an extended class $\mc K_\infty$ function%\footnote{An extended class $\mc K_\infty$ function is a continuous function $\gamma : \R \to \R$ that is strictly increasing and with $\gamma(0) = 0$.}
	$\gamma$ such that, for the system \eqref{eq:ds_syst_cont},
	\begin{equation}
		\label{eq:cbfdefinition}
		\sup_{u \in \mc \R^m}  \left\{L_f h(x) + L_g h(x)u + \gamma(h(x))\right\} \geq 0.
	\end{equation}
	for all $ x \in \mc D$.
\end{definition}
$L_f h(x)$ and $L_g h(x)$ denote the Lie derivatives of $h$ along the vector fields $f$ and $g$. Given this definition of CBFs, the following theorem highlights how they can be used to ensure both set forward invariance (safety) and stability.

\begin{theorem}[Safety and stability \cite{ames2019control}]
	\label{thm:safety}
	Let $\mathcal{S} \subset \mc D\subset \R^n$ be a set defined as the zero superlevel set of a continuously differentiable function $h: \mc D \to \R$. If $h$ is a CBF on $\mc D$ with $0$ a regular value, then any Lipschitz continuous controller $u(x) \in \{ u \in \mc \R^m \colon L_f h(x) + L_g h(x)u + \gamma(h(x)) \geq 0\}$ for the system \eqref{eq:ds_syst_cont} renders the set $\mc S$ forward invariant (safe).  Additionally, the set $\mc S$ is asymptotically stable in $\mc D$.
\end{theorem}

\subsection{Safety Controller Design Using Quadratic Program}

In cases where a nominal controller $\hat u$ is designed to guarantee desired performance, such as stability of an equilibrium point of~\eqref{eq:ds_syst_cont}, a safety filter can be effectively realized by implementing the controller solution of the following \ac{qp}, as shown in \cite{Ames17}:
\begin{equation}
    \begin{split}
    & \minimize_{u} \Vert u - \hat u \Vert^2  \\
    & \st L_f h(x) + L_g h(x)u +\gamma(h(x)) \geq 0.
    \end{split}
\label{eq:qp_safety}
\end{equation}
The objective of the \ac{qp} in~\eqref{eq:qp_safety} is to minimize the difference between $u$ and $\hat u$, while the \ac{cbf} constraint guarantees safety. The condition $L_g h(x) \neq 0$ and the fact that $h$ is a \ac{cbf} ensure that the feasible set is non-empty. 
Thus, the closed-loop system $\dot{x}_{cl} = f(x) + g(x)u^*$, where $u^*$ is the solution of \eqref{eq:qp_safety}, is safe with respect to the set $\mc S$ defined in \eqref{eq:c_set}.

\subsection{Undesirable Asymptotic Equilibria in \texorpdfstring{\ac{cbf}}{}--\texorpdfstring{\acp{qp}}{}}

As far as safety is concerned, when using \acp{cbf}, the geometry of the safe set $\mc S$ does not matter, as long as the latter can be expressed as the zero superlevel set of a continuously differentiable function, as in \eqref{eq:c_set}. Nevertheless, although safety is not undermined, system performance may degrade due to the shape of the safe set. Consider the scenario in which asymptotic stability of a desired equilibrium point of the system is to be achieved by means of a control Lyapunov function. In \cite{Ames17}, the proposed optimization problem features a slack variable which embodies the amount by which asymptotic stability is sacrificed in favor of safety. Moreover, from \eqref{eq:qp_safety}, it is clear that a safe controller is selected only based on the current state $x$. Consequently, it is easy to foresee that in case of non-convex obstacles \textit{deadlocks}---i.e., $u^*$ resulting in the undesired condition $\dot x_{cl}=0$---are more likely to happen.

\begin{figure}
\centering
	\includegraphics[width=0.9\columnwidth]{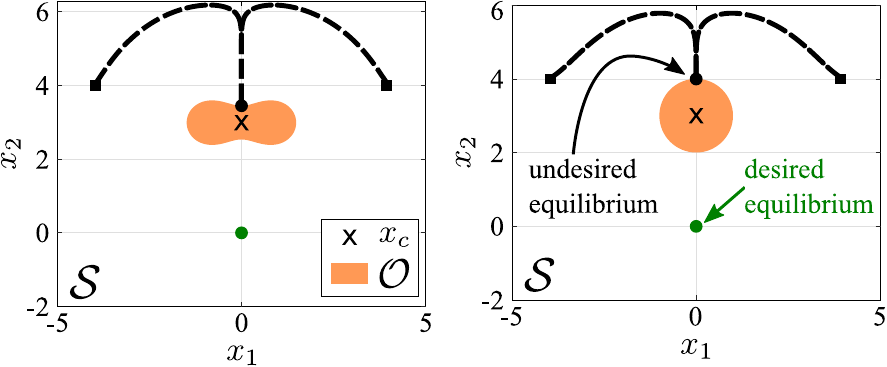}
	\caption{Simulation results showing how \ac{cbf}-based \acp{qp} used to avoid concave (left) and convex (right) obstacles (brown-shaded regions denoted by $\mc O$) generate undesired and asymptotically stable equilibria.}
	\label{fig:undesired_equilibria}
\end{figure}

As an example, consider the case depicted in Fig.~\ref{fig:undesired_equilibria} on the left. The state $x=[x_1,x_2]\tr$ of the dynamical system
\begin{equation}
	\label{eq:dyn_sys_example}
	\dot{x} = - \begin{bmatrix} 6 & 0 \\ 0 & 1 \end{bmatrix} x + u
\end{equation}
is to be kept outside the orange-shaded concave regions. With the following \ac{cbf}
\begin{equation}
\label{eq:cbf_circle}
h(x)= \|x-x_c\|^4 - (x-x_c)\tr\begin{bmatrix}
	10 & 0\\
	0 & -1
\end{bmatrix}(x-x_c),
\end{equation}
where $x_c=[0,3]\tr$, a safe control input $u^*$ is synthesized by solving the \ac{qp} \eqref{eq:qp_safety}, where $\hat u = 0$. The resulting trajectories originating from two different initial conditions are reported. As it can be seen, the funnel-like shape of the unsafe set in the vicinity of $x_c$ results in the closed-loop system to converge to points of the boundary of the safe set, as predicted in \cite{reis2020control}.
	
Though this behavior might seem related to the non-convexity of the obstacle, equilibrium points may arise even with a single convex obstacle. Reis et al.~\cite{reis2020control} have demonstrated that the combination of safety and stability objectives into a \ac{qp} introduces undesirable asymptotically stable equilibrium points in the closed-loop system. In Fig.~\ref{fig:undesired_equilibria} on the right, we show the behavior of the system \eqref{eq:dyn_sys_example} under such conditions. Here, we aim at keeping $x$ outside the circular region of radius $1$ centered at $x_c=[0,3]\tr$, by defining the \ac{cbf} $h(x) = (x - x_c)^2 - 1$.
The resulting trajectories depicted in the figure show how the safe input $u^*$ renders the point $[0,4]\tr$ asymptotically stable, entirely analogously to what happened in the case of non-convex obstacles shown previously. In the next section, we show how to mitigate the issue of spurious asymptotically stable equilibrium points, as well as deadlocks, by reducing the problem only to a zero-measure set of state initial conditions.

%===================== SECTION III ===========================================%
\section{Proposed Approach}
\label{sec:approach}

In order to eliminate undesired asymptotically stable equilibrium points in safe control applications where obstacle regions are non-convex, we follow a three-step approach:
\begin{enumerate}
\item First, we map the state space (real world) into the ball world---a nomenclature introduced in~\cite{rimon1991construction}---using a diffeomorphism $F \colon \R^n \to \R^n$.
\item Then, in the ball world, we develop a \ac{cbf}-based safety controller, which allows the \textit{obstacles to move and shrink in the ball world}, to avoid the image $q$ of the state $x$, realizing what we call a \textit{state-avoidance algorithm}. %, named in analogy with more common obstacle-avoidance algorithms
\item Finally, by re-evaluating the diffeomorphism $F$ based on the updated locations of the obstacles in the ball world, we compute the input $u$ to control the system in the real world.
\end{enumerate}

In Section~\ref{subsec:ballworld}, we define the ball world, which is used to develop the state-avoidance algorithm in Section~\ref{subsec:activeobstacles}. Then, in Section~\ref{subsec:mainqp}, we present an always feasible, computationally efficient implementation of the state-avoidance algorithm which does not introduce asymptotically stable equilibria and leads to deadlocks only for a zero-measure set of initial states.

\subsection{Ball World}
\label{subsec:ballworld}

{The first step in our approach consists of diffeomorphically mapping the compact state space (real world) $\mc X$ to the ball world $\hat{\mc O}_0$}, whose boundary is the $n-1$-dimensional spherical surface $\mathbb S^{n-1}$ and in which unsafe regions, i.e. the \textit{obstacles}, are closed balls. In the ball world, both the state space $\hat{\mc O}_0$ and the obstacles $\hat{\mc O}_i,~i=1,\ldots,M$ are defined as:
\begin{equation}
\label{eq:obstacles}
\hat{\mc O}_i = \{ q \in\R^n \colon \hat \beta_i(q)\le0 \},~i=0,\ldots,M,
\end{equation}
where $\hat \beta_0(q)=\rho_0^2 - \|q_0-q\|^2$, $\hat \beta_i(q)=\|q_i-q\|^2 - \rho_i^2$, for $i\neq0$. $q_i$ and $\rho_i>0$ denote the center and the radius of the $i$-th obstacle, respectively. We use $q$ to denote the state $x$ mapped into the ball world by the diffeomorphism $F$, i.e. $q=F(x)$. Moreover, quantities with a hat in the ball world are used to distinguish them from their analogous in the real world. The safe set in the ball world can be written as
\begin{equation}
	\label{eq:safesetball}
\hat{\mc S} = \{ q \in \R^n \colon \hat\beta_0(q)\ge 0, \ldots, \hat\beta_M(q)\ge 0\}.
\end{equation}

{The main advantage of considering the ball world is that the obstacles are parameterized by their radius and the location of their center, whose rates of change will be controlled to \textit{let the obstacles avoid the state}, achieving, this way, safety. It is worth noticing that simply transforming obstacles into convex regions does not solve the problem of spurious equilibria (see Fig.~\ref{fig:undesired_equilibria} and~\cite{reis2020control}). Therefore, in the rest of this section, we present an approach that reduces the problem of deadlocks to a zero-measure set of initial configurations.
}

\subsection{Safe Control of Obstacles in the Ball World}
\label{subsec:activeobstacles}

The ball world is only the image of the real world through an appropriately defined diffeomorphism. Therefore, it can be deformed and modified to ensure that the real world state never enters obstacle regions. This approach can be interpreted as a state-avoidance paradigm, in which the obstacles are controlled to avoid the state, as opposed to the more traditional obstacle-avoidance paradigm, where a controller $u$ is sought for $x$ in order to avoid the obstacles. In the state-avoidance paradigm, obstacles in the ball world are displaced from their positions, and their radii are changed in order to satisfy the following conditions:
\begin{enumerate}[label=(C\arabic*)]
\item The  state $q$ of the ball world does not collide with the obstacles
\item Obstacles do not collide with each other
\item Obstacles do not leave the ball world
\end{enumerate}

Condition (C1) is directly related to safety: if $q$ is safe in the ball world, then $x$ is safe in the real world. This is because the diffeomorphism $F$ maps the free space in the real world to the free space in the ball world. Notice that considering the outside of the ball world $\hat{\mc O}_0$ as an obstacle, condition (C1) signifies that the state never leaves the ball world. Conditions (C2) and (C3) are introduced to simplify the definition of the diffeomorphism $F$---given in Sec.~\ref{sec:experiments}.

Constraint-satisfaction for active obstacles can be efficiently enforced through the use of \acp{cbf}. In order to do that, we need to define the dynamics of the obstacle motion. In this paper, we choose to control both the center $q_i$ of each obstacle and its radius $\rho_i$ by means of the velocity of the center of the obstacle $u_{q_i}$ and the rate of change of its radius $u_{\rho_i}$, respectively, resulting in the following single integrator dynamical model for each obstacle $i$:
\begin{equation}
\dot q_i = u_{q_i},\quad
\dot \rho_i = u_{\rho_i}.
\end{equation}
In the case of $\hat{\mc O}_0$, we choose to only change its radius, i.e., $\dot q_0 = 0$. As will be shown in the following, this suffices to provably guarantee safety. With these models in place, we are now ready to present the \acp{cbf} employed to enforce conditions (C1), (C2), and (C3).

\paragraph{\ac{cbf} for Condition (C1)} The objective of keeping $q$ away from $\hat {\mc O}_i$ can be enforced by defining a \ac{cbf} $h_i$ based on the expression of $\hat\beta_i$, i.e., $h_i(q_i,\rho_i) = \| q_i - q \|^2 - \rho_i^2$. The differential inequality in \eqref{eq:cbfdefinition} is equivalent in this case to the following affine constraint on $u_{q_i}$ and $u_{\rho_i}$:
\begin{equation}
	\label{eq:cbf_constr_c1}
	A_{C1,i} \begin{bmatrix}
		u_{q_i}\,
		u_{\rho_i}
	\end{bmatrix}\tr \le b_{C1,i}
\end{equation}
where $A_{C1,i}=\begin{bmatrix}
	-2(q_i-q)\tr & 2\rho_i
\end{bmatrix}$ and $b_{C1,i} = {-2(q_i-q)\tr \dot q + \gamma(h_i(q_i,\rho_i))}.$

Analogously, the objective of keeping $q$ within the ball world $\hat {\mc O}_0$ can be enforced by defining the \ac{cbf} $h_0(q_0, \rho_0) = \rho_0^2 - \| q_0 - q \|^2$ which leads to the following constraint on $u_{\rho_0}$:
\begin{equation}
	\label{eq:cbf_constr_c10}
	A_{C1,0} u_{\rho_0} \le b_{C1,0}
\end{equation}
where $A_{C1,0}\!=\!-2\rho_0$ and $b_{C1,0}\!=\!2(q_0-q)\tr \dot q + \gamma(h_0(q_0,\rho_0))$

\paragraph{\ac{cbf} for Condition (C2)} To make the obstacles not collide with each other, we define the following \ac{cbf} for each pair of obstacles $\hat {\mc O}_i$ and $\hat {\mc O}_j$: ${h_{ij}(q_i,q_j,\rho_i,\rho_j) = \| q_i - q_j \|^2 - (\rho_i+\rho_j)^2}$, which lead to the inequality constraint for the inputs to the obstacles:
\begin{equation}
\label{eq:cbf_constr_c2}
A_{C2,ij}
\begin{bmatrix}
u_{q_i}\,
u_{q_j}\, 
u_{\rho_i}\,
u_{\rho_j} 
\end{bmatrix}\tr \le b_{C2,ij},
\end{equation}
with $A_{C2,ij} = [
	-2(q_i-q_j)\tr ~ 2(q_i-q_j)\tr ~ 2(\rho_i+\rho_j) ~ 2(\rho_i+\rho_j)]$ and $b_{C2,ij}=\gamma(h_{ij}(q_i,q_j,\rho_i,\rho_j))$.

\paragraph{\ac{cbf} for Condition (C3)} To satisfy this condition, the obstacles need to be kept within the boundary of the ball world. Proceeding analogously to the previous two cases, we let $h_{i0}(q_i,q_0,\rho_i,\rho_0) = (\rho_0-\rho_i)^2 - \| q_i - q_0 \|^2$ and get the following affine constraint on $u_{q_i}$, $u_{\rho_i}$, and $u_{\rho_0}$:
\begin{equation}
\label{eq:cbf_constr_c3}
A_{C3,i} \begin{bmatrix}
	u_{q_i}\,
	u_{\rho_i}\,
	u_{\rho_0}
\end{bmatrix}\tr \le b_{C3,i}
\end{equation}
where $A_{C3,i}=\begin{bmatrix}
	2(q_i-q_0)\tr & 2(\rho_0-\rho_i) & -2(\rho_0-\rho_i)
\end{bmatrix}$ and $b_{C3,i}=\gamma(h_{i0}(q_i,q_0,\rho_i,\rho_0))$.

\subsection{State-avoidance algorithm}
\label{subsec:mainqp}

In order to synthesize a control input for the obstacles that keeps $q$ safe, satisfying Conditions (C1), (C2), and (C3), we formulate the following \ac{qp}:
\begin{equation}
\label{eq:stateavoidanceQP}
\begin{aligned}
\textbf{Main QP}\hspace{0.8cm}&\\
\minimize_{u_q,u_\rho} &\|u_q-\hat u_q\|^2 + \kappa \|u_\rho-\hat u_\rho\|^2\\
\st & \eqref{eq:cbf_constr_c1}, \eqref{eq:cbf_constr_c10}, \eqref{eq:cbf_constr_c2}, \eqref{eq:cbf_constr_c3}\\
&\quad\forall i,j\in\{1,\ldots, M\},j>i,
\end{aligned}
\end{equation}
where $u_q$ and $u_\rho$ are the stacked obstacle velocity and radius change rate, respectively, i.e.,
\begin{equation}
u_q=[u_{q_1}\tr~\ldots~u_{q_M}\tr]\tr,~~
u_\rho=[u_{\rho_0}~u_{\rho_1}~\ldots~u_{\rho_M}]\tr,
\end{equation}
and $\hat u_q$ and $\hat u_\rho$ are the stacked nominal control inputs for the velocity and the radius change rate of the obstacles. These are defined to keep the originally chosen positions $q_i(t_0)$ and radius $\rho_i(t_0)$---which can be selected arbitrarily in the ball world {, as long as they satisfy conditions (C1) to (C3)}:
\begin{equation}
	\label{eq:obstnom}
\hat u_{q_i} = K_p(q_i(t_0) - q_i),~~
\hat u_{\rho_i} = K_p(\rho_i(t_0) - \rho_i),
\end{equation}
where $K_p>0$ is a controller gain. Finally, the parameter $\kappa$ in \eqref{eq:stateavoidanceQP} is introduced  to favor a position change of the obstacles rather than a size change.

The presented framework is implementable and useful if the Main \ac{qp} is feasible and there are no undesirable asymptotically stable equilibrium points in the resulting state trajectory. These two mentioned properties are the subject of the following two propositions, respectively, which make up the main contribution of this paper.

\begin{proposition}[Feasibility of Main \ac{qp}]
\label{prop:feasibility}
If the system is safe at time $t_0$, the Main \ac{qp} \eqref{eq:stateavoidanceQP} is feasible and will be feasible for all $t> t_0$. Therefore, the system will remain safe for all times $t>t_0$.
\end{proposition}
\begin{proof}
{
We need to show that the feasible set of the Main \ac{qp} \eqref{eq:stateavoidanceQP} is never empty, that is there always exist inputs $u_{q_i}$ and $u_{\rho_i}$ to the obstacles in the feasible set. As the system starts safe, the values of $h_0$, $h_i$, and $h_{ij}$ are positive. Therefore, necessarily $0<\rho_i<\rho_0$ for all $i$. Then, if $\rho_i>0$ $\forall i$, $\exists u_{\rho_i}<0, u_{\rho_0}>0$ s.t. \eqref{eq:cbf_constr_c1} to \eqref{eq:cbf_constr_c3} are all satisfied.

Negative radii rates of change, $u_{\rho_i}<0$, however, might lead to $\rho_i=0$ for some $i$. In this case, \eqref{eq:cbf_constr_c2} and \eqref{eq:cbf_constr_c3} are satisfied by $u_{q_i}=0$, and $\exists u_{\rho_0}>0$ s.t. \eqref{eq:cbf_constr_c10}---which contains the sign-indefinite term $(q_i-q)\tr\dot q$---is satisfied. The constraint \eqref{eq:cbf_constr_c1} also contains the sign-indefinite term $(q_i-q)\tr\dot q$ on the right hand side. Thus, if $(q_i-q)\tr\dot q<0$, i.e., $q$ moves away from  $\hat {\mc O}_i$, then $u_{q_i}=0$ is a solution of \eqref{eq:cbf_constr_c10} as well. Otherwise, a $u_{q_i}\neq0$ so that $(q_i-q)\tr u_{q_j}>0$, i.e.,  $\hat {\mc O}_i$ moves away from $q$, has to be chosen. Such a choice does not conflict with any of the other constraints, because $\exists u_{\rho_0}\neq0, u_{\rho_i}\neq0$ s.t. constraints \eqref{eq:cbf_constr_c1} to \eqref{eq:cbf_constr_c3} are simultaneously satisfied.
}
\end{proof}

Despite the feasibility guarantees given in Proposition~\ref{prop:feasibility}, there could still be cases in which the solution of the Main \ac{qp} leads to the undesired outcome $\dot x=0$ even with $\dot q\neq0$. In the following proposition, we show how the proposed state-avoidance algorithm guarantees that deadlocks only happen for a zero-measure set of configurations, even in the unfavorable presence of concave obstacles in the real world.

\begin{proposition}[Deadlocks introduced by Main \ac{qp}]
\label{prop:deadlocks}
The set of obstacle configurations that may lead to deadlock has measure zero.
\end{proposition}
\begin{proof}
{
In order to quantify the measure of the set of configurations under which deadlock occurs, we can proceed by characterizing these configurations as follows. As the Main \ac{qp} is shown to be feasible in Proposition~\ref{prop:feasibility}, there always exist a control input for the obstacles in the ball world. This means that, the mapped state $q$ in the ball world will move without colliding. As the free space in the ball world is mapped to the free space in the real world by an appropriate diffeomorphism $F$, the only condition under which there may be a deadlock is when $\dot q \neq 0$ in the ball world and yet $\dot x = 0$ in the real world.

Then, assume $\dot q(t) \neq 0$. As $\dot x = \frac{\partial F}{\partial x}^{-1} \dot q$, $\dot x(t)\to0$ if $\left\| \frac{\partial F}{\partial x} \right\| \to\infty$, and since $F$ is a diffeomorphism, this happens only if either (i) the radius of the ball world, $\rho_0$, diverges, or (ii) the radius of obstacle $i$, $\rho_i\to0$, for some $i$. Case (i) only happens when $\forall t$ $u_{q_i}(t) = \alpha (q_i(t)-q_0(t))$, $\alpha>0$, i.e., obstacle $i$ moves towards the boundary of the ball world, whereas case (ii) implies that $\forall t$ $\dot q(t) = \alpha (q_i(t)-q(t))$, $\alpha>0$, i.e., $q$ moves towards the center of obstacle $i$. Therefore, the set of configurations in which (i) or (ii) occur has measure zero in $\R^n$, and, as the number of obstacles is finite, the union of such sets for all obstacles has measure zero.
}
\end{proof}

Proposition~\ref{prop:feasibility}~and~\ref{prop:deadlocks} show that the Main \ac{qp} is feasible and leads to deadlock only for a zero-measure set of configurations, even in the critical scenarios with multiple concave obstacles. This advantage will be showcased in the next section, where the proposed state-avoidance algorithm is compared with a traditional obstacle-avoidance one which does not leverage the mapping to the ball world.

{It is worth mentioning that, by decoupling the system in the real world from the system in the ball world, Theorem 2 in \cite{reis2020control} can be leveraged to obtain a controller for the position and radius of the obstacles which does not introduce undesirable asymptotically stable equilibria when coupled with the \acp{cbf} defined in this section.}

{The described approach is summarized in Algorithm~\ref{alg:safemultipleconcave}.  It is worth mentioning that Step 4 follows from chain rule applied to the derivative of the mapping $q = F(x)$, and it requires the evaluation of the inverse Jacobian of the diffeomorphism, $\frac{\partial {F^{(k)}}^{-1}}{\partial q}$. Moreover, it is important to point out that the evaluation of the control input $u$ in the real world at Step 9 depends on the considered system. For instance, in Sec.~\ref{subsec:robot_experiment}, we apply feedback linearization to a differential-drive robot modeled with unicycle dynamics.}
\begin{algorithm}	
	\caption{{Safety with multiple non-convex obstacles}}
    \label{alg:safemultipleconcave}
	\begin{algorithmic}[1]
		{
		\Require $F$, $\Delta t$, $\gamma$, $\kappa$, $K_p$, $q_i(t_0)$ and $\rho_i(t_0)$, $i=0,\ldots,M$
		\State $k = 0$
		\While{true}
		\State $k \leftarrow k+1$
		\State $\dot q^{(k)} = L_f F^{(k)}\left(x^{(k)}\right) + L_gF^{(k)}\left(x^{(k)}\right)u^{(k)}$
		\State Compute $\hat u_q$ and $\hat u_\rho$ \Comment{\eqref{eq:obstnom}}
		\State Compute $u_q^*$ and $u_\rho^*$ by solving Main QP \Comment{\eqref{eq:stateavoidanceQP}}
		\State $q_i^{(k+1)} \leftarrow q_i^{(k)} + u_{q_i}^*\Delta t$, $\rho_i^{(k+1)} \leftarrow \rho_i^{(k)} + u_{\rho_i}^*\Delta t,~\forall i$
		\State Update $F^{(k+1)}$; compute $\dot x^{(k)} = \dfrac{\partial {F^{(k+1)}}^{-1}}{\partial q}\dot q^{(k)}$
		\State Compute $u^{(k+1)}$ to track $\dot x^{(k)}$; apply $u^{(k+1)}$ \Comment{\eqref{eq:ds_syst_cont}}
		\EndWhile
		}
	\end{algorithmic}
\end{algorithm}

%===================== EXPERIMENTS ===========================================%
\section{Experimental Results}\label{sec:experiments}
In this section, we show the effectiveness of the approach presented in this paper by considering obstacles defined as the union of multiple convex and concave regions of the state space. In order to do that we need a diffeomorphic mapping between the real world and the ball world. In the following, we present a way to construct an example of such a diffeomorphism for concave star-shaped obstacles. {As shown in Fig.~\ref{fig:good_behavior}, the diffeomorphism is then integrated in our \ac{cbf}--\ac{qp} that moves spherical unsafe sets and reduce their radius, guaranteeing safety and convergence in the real world. The obstacle regions need not be star-shaped, nor needs their shape be the same. In these cases, finding an analytic expression of the diffeomorphism between the real and the ball world might not always be possible. Then, one can resort to computational tools capable of providing efficient evaluations of the sought mapping, together with its inverse and Jacobian \cite{loizou2012navigation,choi2015fast}}. 

\subsection{Mapping real world to ball world}\label{subsec:diffeo_ball_sphere}
In our approach, the choice of the diffeomorphism between the state space $\mc X$ and the ball world is arbitrary. In this section, we consider the example of star-shaped concave obstacles, and follow the approach initially described in \cite{rimon1991construction} to construct a diffeomorphism between a real world with star-shaped obstacles and the ball world.

In order to map the real world to the ball word, we start by assuming that the state space $\mc X$ and the unsafe regions within it can be described as super- and sublevel sets, respectively, of real-valued analytic functions $\beta_i \colon \R^n \to \R$ for which 0 is a regular value. The state space $\mc X$ is assumed to be connected and compact such that $\mc X^\circ \subset \{ x\in\R^n \colon \beta_0(x) > 0 \}$ and $\partial \mc X \subset \{ x\in\R^n \colon \beta_0(x) = 0 \}$. Moreover, it is assumed that $M$ obstacles are present in $\mc X$. These are denoted by $\mc O_i$ and correspond to the interior of connected and compact subsets of $\R^n$ such that (i) $\bar {\mc O}_i \subset \mc X^\circ~ \forall i$, (ii) $\mc X \setminus \bar {\mc O}_i \subset \{ x\in\R^n \colon \beta_i(x)>0 \}$, (iii) $\partial {\mc O}_i \subset \{ x\in\R^n \colon \beta_i(x)=0 \}$, and (iv) $\bar {\mc O}_i \cap \bar {\mc O}_j = \emptyset \quad \forall i\neq j$. With these definitions, we can express the safe set as $\mc S = \mc X \setminus \bigcup\limits_{i=1}^M \mc O_i$. As derived in \cite{rimon1991construction}, a diffeomorphism $F$ between the safe set $\mc S$ in the real world and the safe set $\hat {\mc S}$ in the ball world given in \eqref{eq:safesetball} is given by:
    \begin{equation}
	   %\scalebox{0.85}{$
	   F(x) = \sum_{i=0}^{M}\sigma_{i,\lambda}(x)\left( \rho_i f_i(x) + q_i \right)+ \sigma_{g,\lambda}(x) \left(x-x_g + q_g\right),%$}
        \label{eq:diffeo}
    \end{equation}
where $f_i(x) = \frac{\|x-x_i\|}{r_i(\theta)}\begin{bmatrix}
	\cos\theta,
	\sin\theta
\end{bmatrix}\tr,$ with $\theta = \angle(x-x_i)$, $x_i$ is the center of the star-shaped obstacle, and $r_i$ is the function describing how the radius changes as a function of the angle $\theta$. The functions $\sigma_{i,\lambda}$ are defined as 
$		\sigma_{i,\lambda} = \frac{\gamma_g\bar\beta_i}{\gamma_g\bar\beta_i+\lambda \beta_i},~ i=0\ldots M$,  $\sigma_{g,\lambda} =1-\sum_{i=0}^{M} \sigma_{i,\lambda}$, with $\gamma_g = \|x-x_g\|^2$ and $	\bar\beta_i = \prod\limits_{j=0,j\neq i}^{M} \beta_j$. Finally, $x_g$ and $q_g$ are referred to as goal points in the real and ball worlds, respectively. In the context of this work, they can be set to any point in the safe set of the real world and ball world, respectively.

\subsection{Avoid multiple concave regions}
\begin{figure}
\centering
\includegraphics[width=0.8\columnwidth]{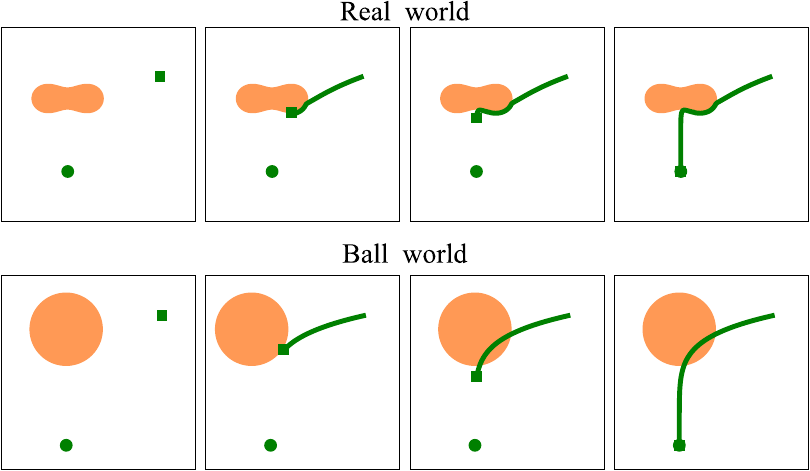}
\caption{{The working principle of our \ac{cbf}--\ac{qp}. In order to avoid a concave obstacle, we first use a diffeomorphism to map the obstacle (orange shapes in real world) into a sphere (orange shapes in ball world). We then use our \ac{cbf}--\ac{qp} to move (and reduce the radius---not shown) of the obstacle in the sphere world. The solution of the \ac{cbf}--\ac{qp} is mapped back to the real word to render the safe set forward invariant. }} 
\label{fig:good_behavior}
\end{figure}

In this experiment, we consider multiple non-convex obstacles (depicted in orange in Fig.~\ref{fig:avoid_concave}). The unsafe set is made up of star-shape concave obstacle regions described as detailed in Fig.~\ref{fig:avoid_concave}. The uncontrolled dynamics are given in~\eqref{eq:dyn_sys_example}, and are stable but unsafe. The results of the application of our proposed method are compared with the traditional \ac{cbf}--\ac{qp} \eqref{eq:qp_safety}. Figure~\ref{fig:avoid_concave} shows that the system without any CBF constraint (blue dotted lines) is stable but not safe, while the controller synthesized using a standard \ac{cbf}--\ac{qp} (black dashed lines) generates a stable attractor at the boundary of the unsafe set. For the standard \ac{cbf}--\ac{qp}, we used the \ac{cbf} in~\eqref{eq:cbf_circle}. The proposed Main QP generates trajectories (green solid lines) which are both safe and asymptotically converge to the origin. The parameters of the diffeomorphism \eqref{eq:diffeo} and the Main QP \eqref{eq:stateavoidanceQP} have been set to $\lambda = 100$ and $\kappa=1$. 

\begin{figure}
\centering
    \includegraphics[width=0.8\columnwidth]{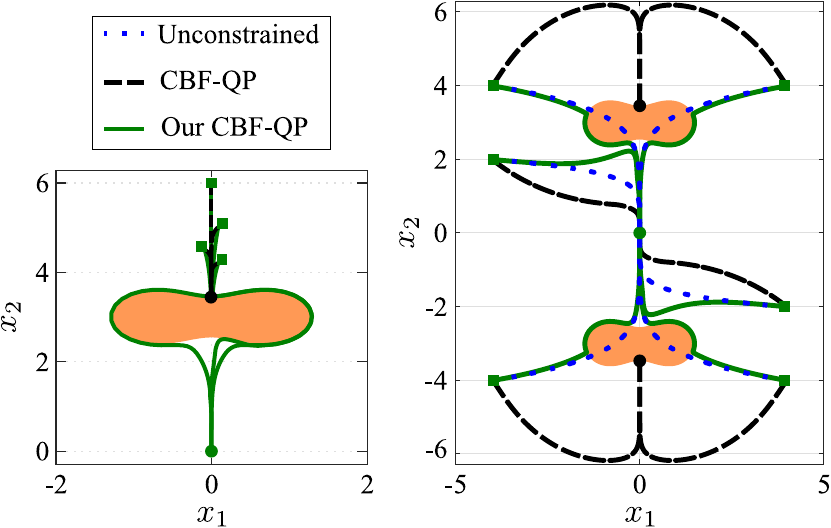}
    \caption{\ac{cbf}-based \acp{qp} used to avoid single (left) and multiple (right) concave obstacle regions (orange shapes) starting from different initial conditions. {Obstacles are described as $\beta_i(x) = \big((x_1-x_{ci,1}-a)^2+(x_2-x_{ci,2})^2\big) \big((x_1-x_{ci,1}+a)^2 + (x_2-x_{ci,1})^2\big) - b^4$ where $i=1,2$, $a=1$, and $b=1.1$. The centers are placed at $x_{c1}=[0,3]\tr$ and $x_{c2}=[0,-3]\tr$.} The unconstrained system dynamics is unsafe as it enters the unsafe set. The standard \ac{cbf}--\ac{qp} ensures safety, but it generates an undesired asymptotically stable equilibrium point. Our \ac{cbf}--\ac{qp} ensures both safety and stability to the desired equilibrium point, i.e., the origin. {Our \ac{cbf}--\ac{qp} leads to deadlock only for a zero-measure set of initial configurations---the line $\{(x_1,x_2)\colon x_1=0, x_2\geq 3.45\}$ in the left panel.}}
     \label{fig:avoid_concave}
\end{figure}

Proposition~\ref{prop:deadlocks} guarantees that our approach lead to deadlock only for a zero-measure set of configurations. We experimentally verified this proposition and the results are shown in Fig.~\ref{fig:avoid_concave} (left). Deadlocks occur only for initial states starting from a line (zero-measure set). For any other initial state, arbitrarily close to this line, the stability of the origin is guaranteed. As already discussed, in the same situation a traditional \ac{cbf}--\ac{qp} generates an undesired asymptotically stable equilibrium point.

\subsection{{Robotic Experiments}}\label{subsec:robot_experiment}
\begin{figure}
\centering
\subfigure[]{\label{subfig:robot_arm}\includegraphics[width=0.24\textwidth]{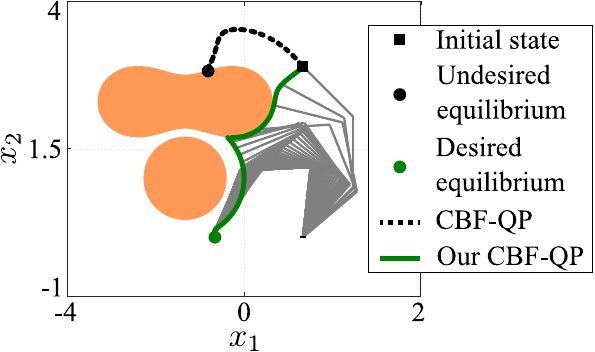}}~
\subfigure[]{\label{subfig:robotarium_with}\includegraphics[trim={3cm 0cm 8cm 4cm},clip,width=0.23\textwidth]{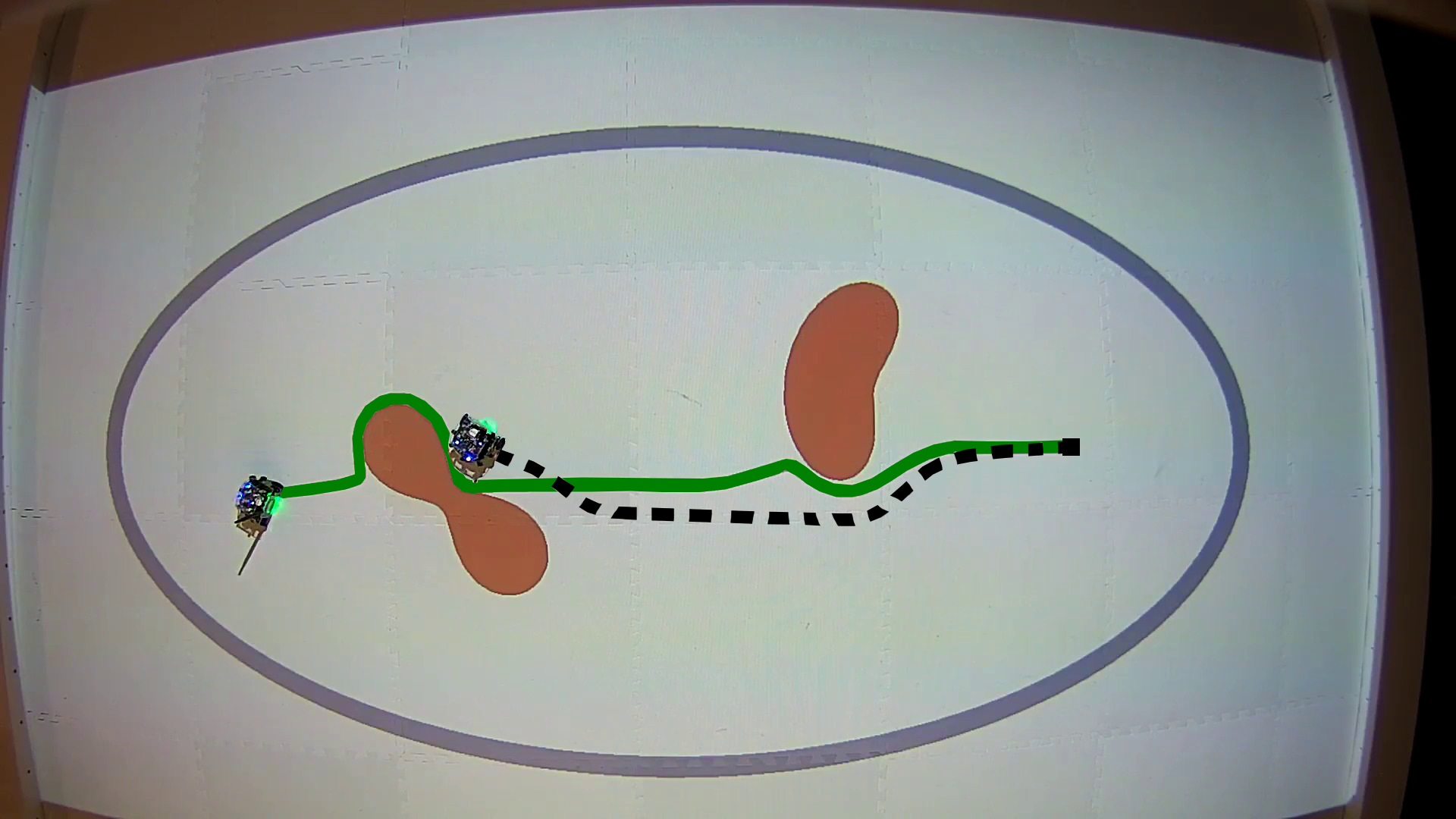}}
\caption{Two robotic experiments performed using traditional and proposed \ac{cbf}--\ac{qp}. With the proposed approach both (a) a planar 3-link manipulator (grey lines) and (b) a mobile robot navigating in the Robotarium are able to reach the goal (green dot) while avoiding multiple unsafe regions (orange-shaded regions). 
\changed{A full video of the Robotarium experiment is available at } \protect\url{https://youtu.be/p18XnJShUGU}.}
\label{fig:robotarium}
\vspace{-0.2cm}
\end{figure}

The presented approach has been implemented on a simulated 3-link planar arm and a real differential-drive robot on the Robotarium \cite{wilson2020robotarium}. Figure~\ref{fig:robotarium} shows the two experiments performed using \changed{traditional and our \ac{cbf}--\ac{qp}}. In Fig.~\ref{subfig:robot_arm}, with our formulation the end-effector reaches the goal while avoiding 2 obstacles, while the traditional \ac{cbf}--\ac{qp} generates an undesired equilibrium. In Fig.~\ref{subfig:robotarium_with}. An elliptical environment containing two obstacles, (projected onto the Robotarium testbed as orange regions) has to be navigated by the robot to reach the goal location (green dot on the left). Also in this case, the traditional \ac{cbf}--\ac{qp} generates an undesired equilibrium point, which is eliminated using the proposed approach.

%===================== CONCLUSION ===========================================%
\section{Conclusions}\label{sec:conclusion}
In this paper, we presented a method to ensure safety of dynamical systems in presence of multiple non-convex obstacle regions. The developed algorithm leverages control barrier functions and it is amenable for fast and online implementation in many feedback control applications. Moreover, despite the reactive nature of the developed controller, the closed-loop system does not suffer from the presence of undesirable asymptotically stable equilibrium points, typically generated when competing objectives of safety and stability are enforced using control barrier functions. The effectiveness of the proposed approach is shown in simulation and through its implementation on a real robotic platform employed in a navigation scenario.

%===================== bibliography ===========================================%
\bibliographystyle{IEEEtran}
\bibliography{bibliography.bib}

\end{document}